\documentclass[a4paper,11pt]{article}
\usepackage{amsmath,amsthm,amssymb}
\usepackage{graphicx}
\usepackage{mathrsfs}
\usepackage{CJKutf8}
\newtheorem{theorem}{Theorem}[section]
\newtheorem*{theorem*}{Theorem}

\newtheorem{proposition}[theorem]{Proposition}
\newtheorem{corollary}[theorem]{Corollary}
\newtheorem*{corollary*}{Corollary}
\newtheorem{lemma}[theorem]{Lemma}
\newtheorem{example}[theorem]{Example}
\newtheorem{remark}[theorem]{Remark}

\newcommand{\BZ}{\mathbb{Z}}
\newcommand{\BR}{\mathbb{R}}
\newcommand{\BC}{\mathbb{C}}
\newcommand{\bA}{\mathbf{A}}
\newcommand{\bB}{\mathbf{B}}
\newcommand{\bC}{\mathbf{C}}
\newcommand{\bD}{\mathbf{D}}
\newcommand{\bH}{\mathbf{H}}
\newcommand{\bI}{\mathbf{I}}
\newcommand{\bJ}{\mathbf{J}}
\newcommand{\bP}{\mathbf{P}}
\newcommand{\bQ}{\mathbf{Q}}
\newcommand{\be}{\mathbf{e}}
\newcommand{\bm}{\mathbf{m}}
\newcommand{\bn}{\mathbf{n}}
\newcommand{\bsigma}{\boldsymbol{\sigma}}
\newcommand{\bzero}{\mathbf{0}}
\newcommand{\cB}{\mathcal{B}}
\newcommand{\cC}{\mathcal{C}}
\newcommand{\cF}{\mathcal{F}}
\newcommand{\cO}{\mathcal{O}}
\newcommand{\cH}{\mathcal{H}}
\newcommand{\cL}{\mathcal{L}}
\newcommand{\cP}{\mathcal{P}}

\newcommand{\ra}{\mathrm{a}}
\newcommand{\Int}{\operatorname{Int}}
\newcommand{\Ad}{\operatorname{Ad}}
\newcommand{\adj}{\operatorname{adj}}
\newcommand{\rank}{\operatorname{rank}}
\newcommand{\Ker}{\operatorname{Ker}}
\newcommand{\tr}{\operatorname{tr}}

\renewcommand{\Im}{\operatorname{Im}}
\newcommand{\hsum}{\sideset{}{^\oplus}{\sum}}

\newcommand{\End}{\operatorname{End}}

\allowdisplaybreaks[3]

\title{Representation theory of $\mathfrak{sl}(2,\BR)\simeq \mathfrak{su}(1,1)$ and \\ a generalization of non-commutative harmonic oscillators}
\author{Ryosuke Nakahama\thanks{This work was supported by JST CREST Grant Number JPMJCR2113, Japan.}
\thanks{Email: ryosuke.nakahama@ntt.com} \\
\textit{NTT Institute for Fundamental Mathematics,} \\ \textit{NTT Communication Science Laboratories,} \\ \textit{Nippon Telegraph and Telephone Corporation,} \\
\textit{3-9-11 Midori-cho, Musashino-shi, Tokyo 180-8585, Japan}}
\date{\today}

\begin{document}

\maketitle

\begin{abstract}
The non-commutative harmonic oscillator (NCHO) was introduced as a specific Hamiltonian operator on $L^2(\BR)\otimes\BC^2$ by Parmeggiani and Wakayama. 
Then it was proved by Ochiai and Wakayama that the eigenvalue problem for NCHO is reduced to a Heun differential equation. 
In this article, we consider some generalization of NCHO for $L^2(\BR^n)\otimes\BC^p$ as a rotation-invariant differential equation. 
Then by applying a representation theory of $\mathfrak{sl}(2,\BR)\simeq \mathfrak{su}(1,1)$, 
we check that its restriction to the space of products of radial functions and homogeneous harmonic polynomials is reduced to 
a holomorphic differential equation on the unit disk, which is generically Fuchsian. \bigskip

\end{abstract}

\section{Introduction}

In \cite{PW1, PW2}, Parmeggiani and Wakayama introduced the non-commutative harmonic oscillator (NCHO) on $L^2(\BR)\otimes\BC^2$. 
For the $\eta$-shifted version of the NCHO \cite{RW}, its Hamiltonian is given by the self-adjoint operator 
\begin{align*}
H_{\text{NCHO}}:\hspace{-3pt}&=\bA\biggl(-\frac{1}{2}\frac{d^2}{dx^2}+\frac{1}{2}x^2\biggr)+\bJ\biggl(x\frac{d}{dx}+\frac{1}{2}\biggr)+2\eta i\det(\bA\pm i\bJ)^{1/2}\bJ, 
\end{align*}
where $\bA,\bJ\in M(2,\BR)$, $\bA={}^t\hspace{-1pt}\bA$, $\bJ=-{}^t\hspace{-1pt}\bJ$, and $\eta\in\BR$. The original NCHO corresponds to $\eta=0$ case. 
Then in \cite{O1, O2, RW, W}, Ochiai, Reyes-Bustos and Wakayama proved that its eigenvalue problem $H_{\text{NCHO}}\psi=\lambda \psi$ is reduced to a Heun differential equation 
on a certain simply-connected domain in $\BC$. In the proof they used (quasi) intertwining operators for (reducible) principal series representations of $\mathfrak{sl}(2,\BR)$. 

In this article, we consider a generalization of this problem, that is, we want to find a function $\psi\in L^2(\BR^n)\otimes\BC^p$ satisfying 
\begin{equation}\label{formula_NCHO_intro}
\biggl[-\bA_1\sum_{j=1}^n\frac{\partial^2}{\partial x_j^2}+i\bA_2\sum_{j=1}^n\biggl(x_j\frac{\partial}{\partial x_j}+\frac{1}{2}\biggr)+\bA_3\sum_{j=1}^n x_j^2\biggr]\psi=2\bC \psi, 
\end{equation}
where $\bA_k,\bC\in M(p,\BC)$, $\bA_k=\bA_k^\dagger$, $\bC=\bC^\dagger$. Here $\bC^\dagger={}^t\hspace{-1pt}\overline{\bC}$. 
This equation commutes with the rotation of $\BR^n$ by orthogonal matrices $g\in O(n)$. 
Then according to the decomposition of $L^2(\BR^n)$ under $\mathfrak{sl}(2,\BR)\times O(n)$ in terms of harmonic polynomials, 
\[ L^2(\BR^n)\simeq \hsum_{k\in\BZ_{\ge 0}} L^2_{k+\frac{n}{2}}(\BR^+)\otimes \cH\cP_k(\BR^n) \]
(where notations are explained later), we check that the restriction of (\ref{formula_NCHO_intro}) on $L^2_{k+\frac{n}{2}}(\BR^+)\otimes \cH\cP_k(\BR^n)\otimes\BC^p$ is 
reduced to a holomorphic differential equation on the unit disk $\bD$, which has some $SU(1,1)$-covariance and is generically Fuchsian. 
As a tool we use intertwining operators for unitary highest weight representations of $\mathfrak{sl}(2,\BR)$, between $L^2_{k+\frac{n}{2}}(\BR^+)$ and the space of holomorphic functions on $\bD$. 

This paper is organized as follows. In Section \ref{section_HO} we begin with a review on usual harmonic oscillators. 
In Section \ref{section_sl2} we review representation theory for the Lie algebra $\mathfrak{sl}(2,\BR)\simeq\mathfrak{su}(1,1)$. 
In Section \ref{section_NCHO} we apply representation theory of $\mathfrak{sl}(2,\BR)\simeq\mathfrak{su}(1,1)$ for the analysis of the generalized NCHO.

\section{Usual harmonic oscillator and harmonic polynomials}\label{section_HO}

First we recall the usual harmonic oscillator of $n$-variables. For details about this and the next section, see, e.g., \cite{F, HT, N}. 
Let $H_{\text{HO}}$ be the self-adjoint operator on $L^2(\BR^n)$ given by 
\[ H_{\text{HO}}:=\frac{1}{2}\sum_{j=1}^n\biggl(-\frac{\partial^2}{\partial x_j^2}+x_j^2\biggr). \]
To find eigenfunctions of $H_{\text{HO}}$, we consider the creation and annihilation operators  
\[ \ra_j^\dagger:=\frac{1}{\sqrt{2}}\biggl(x_j-\frac{\partial}{\partial x_j}\biggr), \qquad \ra_j:=\frac{1}{\sqrt{2}}\biggl(x_j+\frac{\partial}{\partial x_j}\biggr), \]
so that 
\[ H_{\text{HO}}=\sum_{j=1}^n\ra_j^\dagger \ra_j+\frac{n}{2}, \qquad [\ra_j,\ra_k^\dagger]=\ra_j\ra_k^\dagger-\ra_k^\dagger \ra_j=\delta_{jk} \]
hold. Also let 
\[ h_0(x):=\pi^{-n/4}e^{-\sum_{j=1}^n x_j^2/2}\in L^2(\BR^n), \] 
and for $\bm=(m_1,\ldots,m_n)\in(\BZ_{\ge 0})^n$ let 
\[ h_\bm(x):=\frac{1}{\sqrt{\bm!}}(\ra_1^\dagger)^{m_1}\cdots (\ra_n^\dagger)^{m_n}h_0(x)
=\frac{\pi^{-n/4}}{\sqrt{\bm!}}\biggl(-\frac{1}{\sqrt{2}}\biggr)^{|\bm|}\prod_{j=1}^n e^{x_j^2/2}\biggl(\frac{\partial}{\partial x_j}\biggr)^{m_j}e^{-x_j^2}, \]
where $\bm!:=m_1!\cdots m_n!$, $|\bm|:=m_1+\cdots+m_n$. These $h_\bm(x)$ are called the Hermite functions. 
Then we can show 
\begin{gather*}
\ra_j^\dagger h_\bm(x)=\sqrt{m_j+1}h_{\bm+\be_j}(x), \qquad \ra_j h_\bm(x)=\begin{cases} \sqrt{m_j}h_{\bm-\be_j}(x) & (m_j\ge 1), \\ 0 & (m_j=0), \end{cases} \\
\langle h_\bm(x),h_\bn(x)\rangle_{L^2(\BR^n)}=\int_{\BR^n}h_\bm(x)\overline{h_\bn(x)}\,dx=\delta_{\bm\bn}, 
\end{gather*}
where $\be_j=(0,\ldots,0,\overset{j\text{th}}{\check{1}},0,\ldots,0)$. Especially $h_\bm(x)$ is an eigenfunction of $H_{\text{HO}}$ with the eigenvalue $|\bm|+\frac{n}{2}$. 
Now since 
\[ L^2(\BR^n)_{\text{fin}}:=\bigoplus_{\bm\in(\BZ_{\ge 0})^n}\BC h_\bm(x)=\cP(\BR^n)h_0(x) \]
is dense in $L^2(\BR^n)$, where $\cP(\BR^n)$ is the space of polynomials on $\BR^n$, the set of eigenvalues of $H_{\text{HO}}$ is $\BZ_{\ge 0}+\frac{n}{2}$, 
and for each $k\in\BZ_{\ge 0}$, the corresponding eigenspaces are given by 
\[ \biggl\{ \psi\in L^2(\BR^n)\biggm|H_{\text{HO}}\psi=\biggl(k+\frac{n}{2}\biggr)\psi\biggr\}=\bigoplus_{\substack{\bm\in(\BZ_{\ge 0})^n \\ |\bm|=k}}\BC h_\bm(x). \]
The set $\{h_\bm(x)\}_{\bm\in(\BZ_{\ge_0})^n}$ forms an orthonormal basis of $L^2(\BR^n)$. 

The space $L^2(\BR^n)$ is unitarily equivalent to the following \textit{Fock space}
\[ \cF(\BC^n):=\biggl\{ f\in\cO(\BC^n)\biggm| \Vert f\Vert_{\cF(\BC^n)}^2:=\frac{1}{\pi^n}\int_{\BC^n}|f(w)|^2e^{-\sum_{j=1}^n|w_j|^2}\,dw<\infty \biggr\}, \]
where $\cO(\BC^n)$ denotes the space of holomorphic functions on $\BC^n$, via the \textit{Bargmann transform} 
\begin{gather*}
\cB\colon L^2(\BR^n)\longrightarrow \cF(\BC^n), \\
(\cB \psi)(w):=\pi^{-n/4}\int_{\BR^n}\psi(x) e^{\sum_{j=1}^n\bigl(\sqrt{2}x_jw_j-\frac{x_j^2}{2}-\frac{w_j^2}{2}\bigr)}\,dx. 
\end{gather*}
Via $\cB$, the creation and annihilation operators are transformed as 
\[ (\cB(\ra_j^\dagger \psi))(w)=w_j(\cB \psi)(w), \qquad (\cB(\ra_j \psi))(w)=\frac{\partial}{\partial w_j}(\cB \psi)(w). \]
Especially we have 
\begin{gather*}
\cB\circ H_{\text{HO}}\circ\cB^{-1}=\cB\circ\biggl(\sum_{j=1}^n\ra_j^\dagger \ra_j+\frac{n}{2}\biggr)\circ\cB^{-1}=\sum_{j=1}^nw_j\frac{\partial}{\partial w_j}+\frac{n}{2}, \\
(\cB h_\bm)(w)=\frac{1}{\sqrt{\bm!}}w_1^{m_1}\cdots w_n^{m_n}. 
\end{gather*}
The structure of eigenspaces of $\sum_{j=1}^nw_j\frac{\partial}{\partial w_j}+\frac{n}{2}$ coincides with that of $H_{\text{HO}}$. 

Next we recall harmonic polynomials. For $k\in\BZ_{\ge 0}$, let $\cP_k(\BR^n)$ be the space of homogeneous polynomials of degree $k$, 
and let $\cH\cP_k(\BR^n)$ be the space of homogeneous harmonic polynomials of degree $k$. 
\begin{align*}
\cP_k(\BR^n)&:=\{ h(x)\in\cP(\BR^n)\mid h(tx)=t^kh(x) \; (t\in\BR)\}, \\
\cH\cP_k(\BR^n)&:=\biggl\{ h(x)\in\cP_k(\BR^n)\biggm| \sum_{j=1}^n\frac{\partial^2h}{\partial x_j^2}(x)=0\biggr\}. 
\end{align*}
We note that if $n=1$, then $\cH\cP_0(\BR)=\BC$, $\cH\cP_1(\BR)=\BC x$ and $\cH\cP_k(\BR)=\{0\}$ for $k\ge 2$. 
Then it is well-known that we have 
\[ \cP_k(\BR^n)=\cH\cP_k(\BR^n)\oplus\cP_{k-2}(\BR^n)r^2, \]
where $r^2(x):=x_1^2+\cdots+x_n^2$, and hence we have 
\begin{align*}
\cP(\BR^n)&=\bigoplus_{k=0}^\infty \BC[r^2]\cH\cP_k(\BR^n), \\
L^2(\BR^n)_{\text{fin}}&=\bigoplus_{k=0}^\infty \BC[r^2]\cH\cP_k(\BR^n)e^{-r^2/2}, 
\end{align*}
where $\BC[r^2]:=\bigoplus_{j=0}^\infty \BC r^{2j}$. Also, for $\phi(t)\in\BC[t]e^{-t}$, $h(x)\in\cH\cP_k(\BR^n)$, we have 
\begin{align*}
&\int_{\BR^n}\biggl|\phi\biggl(\frac{r^2(x)}{2}\biggr)h(x)\biggr|^2\,dx
=\iint_{\BR^+\times S^{n-1}}\biggl|\phi\biggl(\frac{r^2}{2}\biggr)h(r\omega)\biggr|^2 r^{n-1}\,drd\omega \\
&=\int_0^\infty \biggl|\phi\biggl(\frac{r^2}{2}\biggr)\biggr|^2 r^{2k+n-1}\,dr \int_{S^{n-1}}|h(\omega)|^2\,d\omega \\
&=2^{k+\frac{n}{2}-1}\int_0^\infty |\phi(t)|^2 t^{k+\frac{n}{2}-1}\,dt \int_{S^{n-1}}|h(\omega)|^2\,d\omega
\cdot\frac{2}{\Gamma\bigl(k+\frac{n}{2}\bigr)}\int_0^\infty e^{-s^2}s^{2k+n-1}\,ds \\
&=\frac{2^{k+\frac{n}{2}}}{\Gamma\bigl(k+\frac{n}{2}\bigr)}\int_0^\infty |\phi(t)|^2 t^{k+\frac{n}{2}-1}\,dt \iint_{\BR^+\times S^{n-1}}|h(s\omega)|^2 e^{-s^2}s^{n-1}\,dsd\omega \\
&=\frac{2^{k+\frac{n}{2}}}{\Gamma\bigl(k+\frac{n}{2}\bigr)}\int_0^\infty |\phi(t)|^2 t^{k+\frac{n}{2}-1}\,dt \int_{\BR^n}|h(x)|^2 e^{-r^2(x)}\,dx, 
\end{align*}
where $\BR^+:=\{t\in\BR\mid t>0\}$, $S^{n-1}:=\{x\in\BR^n\mid r^2(x)=1\}$. Now, for $\mu>0$ we define a Hilbert space $L^2_\mu(\BR^+)$ by 
\[ L^2_\mu(\BR^+):=\biggl\{ \phi\colon\BR^+\longrightarrow\BC \biggm| \Vert \phi\Vert_{\mu,\BR^+}^2:=\frac{2^\mu}{\Gamma(\mu)}\int_0^\infty |\phi(t)|^2t^{\mu-1}\,dt<\infty\biggr\}, \]
and define an inner product on $\cH\cP_k(\BR^n)$ by 
\[ \langle h_1,h_2\rangle_{\cH\cP}:=\int_{\BR^n}h_1(x)\overline{h_2(x)}e^{-r^2(x)}\,dx. \]
Then we have a Hilbert direct sum decomposition 
\begin{equation}\label{formula_decomp}
L^2(\BR^n)\simeq \hsum_{k\in\BZ_{\ge 0}} L^2_{k+\frac{n}{2}}(\BR^+)\otimes \cH\cP_k(\BR^n). 
\end{equation}
Especially, if $n=1$ then $L^2(\BR)=L^2_{\text{even}}(\BR)\oplus L^2_{\text{odd}}(\BR)\simeq L^2_{1/2}(\BR^+)\oplus L^2_{3/2}(\BR^+)x$ holds. 
In addition, for $\phi(t)\in L^2_{k+\frac{n}{2}}(\BR^+)_{\text{fin}}:=\BC[t]e^{-t}\subset L^2_{k+\frac{n}{2}}(\BR^+)$, $h(x)\in \cH\cP_k(\BR^n)$, we have 
\[ H_{\text{HO}}\biggl(\phi\biggl(\frac{r^2(x)}{2}\biggr)h(x)\biggr)=h(x)\biggl(-t\frac{d^2}{dt^2}-\biggl(k+\frac{n}{2}\biggr)\frac{d}{dt}+t\biggr)\phi(t)\biggr|_{t=r^2(x)/2}. \]
That is, the eigenvalue problem for $H_{\text{HO}}$ on $L^2(\BR^n)$ restricted to $L^2_{k+\frac{n}{2}}(\BR^+)\otimes \cH\cP_k(\BR^n)$ is equivalent to that for 
\begin{equation}\label{formula_HO_nu}
H_{\text{HO}}^{(\mu)}:=-t\frac{d^2}{dt^2}-\mu\frac{d}{dt}+t 
\end{equation}
on $L^2_\mu(\BR^+)$ for $\mu=k+\frac{n}{2}$. We note that the orthogonal group $O(n):=\{g\in M(n,\BR)\mid {}^t\hspace{-1pt}gg=I_n\}$ acts irreducibly on $\cH\cP_k(\BR^n)$ by $h(x)\mapsto h(g^{-1}x)$. 
In the next section we review the $\mathfrak{sl}(2,\BR)$-module structure on $L^2_\mu(\BR^+)$. 
Especially, (\ref{formula_decomp}) gives a decomposition as representations of $\mathfrak{sl}(2,\BR)\times O(n)$, which is a special case of \textit{Howe's dual pair correspondence} (see \cite{H, KV}).

\section{Representation theory of $\mathfrak{sl}(2,\BR)\simeq\mathfrak{su}(1,1)$}\label{section_sl2}

In this section we consider the Lie groups 
\begin{align*}
SL(2,\BR)&:=\biggl\{g=\begin{bmatrix} a&b\\c&d \end{bmatrix} \biggm| \begin{matrix} a,b,c,d\in\BR, \\ \det(g)=ad-bc=1 \end{matrix} \biggr\}, \\
SU(1,1)&:=\biggl\{g=\begin{bmatrix} a&b\\\overline{b}&\overline{a} \end{bmatrix} \biggm| \begin{matrix} a,b\in\BC, \\ \det(g)=|a|^2-|b|^2=1 \end{matrix} \biggr\}. 
\end{align*}
Then their corresponding Lie algebras are given by 
\begin{align*}
\mathfrak{sl}(2,\BR):=\biggl\{\begin{bmatrix} a&b\\c&-a \end{bmatrix} \biggm| a,b,c\in\BR \biggr\}, \qquad
\mathfrak{su}(1,1):=\biggl\{\begin{bmatrix} ia&b\\\overline{b}&-ia \end{bmatrix} \biggm| a\in\BR,\, b\in\BC\biggr\}. 
\end{align*}
That is, for $X,Y\in \mathfrak{sl}(2,\BR)$ (resp. $\mathfrak{su}(1,1)$), $[X,Y]:=XY-YX\in \mathfrak{sl}(2,\BR)$ (resp. $\mathfrak{su}(1,1)$) holds, 
and $e^{tX}=\sum_{n=0}^\infty (tX)^n/n!\in SL(2,\BR)$ (resp. $SU(1,1)$) holds for all $t\in\BR$. Here $i:=\sqrt{-1}$. 
Their complexification is given by 
\[ \mathfrak{sl}(2,\BR)\otimes_\BR\BC=\mathfrak{su}(1,1)\otimes_\BR\BC=\mathfrak{sl}(2,\BC):=\biggl\{\begin{bmatrix} a&b\\c&-a \end{bmatrix} \biggm| a,b,c\in\BC \biggr\}. \]
Let $C:=\bigl[\begin{smallmatrix}1&i\\i&1\end{smallmatrix}\bigr]$. Then we have isomorphisms 
\begin{align*}
&\Int(C)\colon SU(1,1)\longrightarrow SL(2,\BR), & g&\longmapsto CgC^{-1}, \\
&\Ad(C)\colon \mathfrak{su}(1,1)\longrightarrow \mathfrak{sl}(2,\BR), & X&\longmapsto CXC^{-1}. 
\end{align*}
We take two bases $\{H,E,F\},\{{}^c\hspace{-1pt}H,{}^c\hspace{-1pt}E,{}^c\hspace{-1pt}F\}\subset\mathfrak{sl}(2,\BC)$ by 
\begin{align*}
H&:=\begin{bmatrix}1&0\\0&-1\end{bmatrix}, &E&:=\begin{bmatrix}0&1\\0&0\end{bmatrix}, &F&:=\begin{bmatrix}0&0\\1&0\end{bmatrix}, \\
{}^c\hspace{-1pt}H&:=\Ad(C)H, &{}^c\hspace{-1pt}E&:=\Ad(C)E, &{}^c\hspace{-1pt}F&:=\Ad(C)F, 
\end{align*}
so that 
\begin{align*}
[H,E]&=2E, &[H,F]&=-2F, &[E,F]&=H, \\ 
[{}^c\hspace{-1pt}H,{}^c\hspace{-1pt}E]&=2{}^c\hspace{-1pt}E, &[{}^c\hspace{-1pt}H,{}^c\hspace{-1pt}F]&=-2{}^c\hspace{-1pt}F, &[{}^c\hspace{-1pt}E,{}^c\hspace{-1pt}F]&={}^c\hspace{-1pt}H. 
\end{align*}
A representation $(\tau,V)$ of $\mathfrak{sl}(2,\BC)$ is a pair of a $\BC$-vector space $V$ and a $\BC$-linear map 
$\tau\colon\mathfrak{sl}(2,\BC)\rightarrow \End_\BC(V):=\{\BC\text{-linear maps }V\to V\}$ satisfying $\tau([X,Y])=[\tau(X),\tau(Y)]:=\tau(X)\tau(Y)-\tau(Y)\tau(X)$ for all $X,Y\in\mathfrak{sl}(2,\BC)$. 
We say $(\tau,V)$ is $\mathfrak{sl}(2,\BR)$-(infinitesimally) unitary (resp. $\mathfrak{su}(1,1)$-(infinitesimally) unitary) if there exists an inner product $\langle\cdot,\cdot\rangle$ on $V$ 
satisfying $\langle\tau(X)v,w\rangle=-\langle v,\tau(X)w\rangle$ for all $v,w\in V$ and $X\in\mathfrak{sl}(2,\BR)$ (resp. $\mathfrak{su}(1,1)$), or equivalently, 
\begin{align*}
&\langle\tau(H)v,w\rangle=-\langle v,\tau(H)w\rangle, \quad \langle\tau(E)v,w\rangle=-\langle v,\tau(E)w\rangle, \quad 
\langle\tau(F)v,w\rangle=-\langle v,\tau(F)w\rangle, 
\end{align*}
resp. 
\begin{align*}
\langle\tau(H)v,w\rangle=\langle v,\tau(H)w\rangle, \qquad \langle\tau(E)v,w\rangle=-\langle v,\tau(F)w\rangle. 
\end{align*}
In the following we give some examples of representations of $\mathfrak{sl}(2,\BC)$. 

First we consider the spaces of holomorphic functions on the unit disk and the upper half plane. Let 
\begin{align*}
\bD:=\{z\in\BC\mid |z|<1\}, \qquad \bH:=\{y\in\BC\mid \Im y>0\}, 
\end{align*}
and for $\mu>1$, let $\cH_\mu(\bD)$, $\cH_\mu(\bH)$ be the Hilbert spaces given by 
\begin{align*}
\cH_\mu(\bD)&:=\biggl\{ f(z)\in\cO(\bD)\biggm| \Vert f\Vert_{\mu,\bD}^2:=\frac{\mu-1}{\pi}\int_{\bD}|f(z)|^2(1-|z|^2)^{\mu-2}\,dz<\infty\biggr\}, \\
\cH_\mu(\bH)&:=\biggl\{ f(y)\in\cO(\bH)\biggm| \Vert f\Vert_{\mu,\bH}^2:=\frac{\mu-1}{4\pi}\int_{\bH}|f(y)|^2(\Im y)^{\mu-2}\,dy<\infty\biggr\}. 
\end{align*}
These spaces are called \textit{weighted Bergman spaces}. These are unitarily isomorphic by 
\begin{align*}
\cC_\mu&\colon\cH_\mu(\bD)\longrightarrow\cH_\mu(\bH), & (\cC_\mu f)(y)&:=\biggl(\frac{-iy+1}{2}\biggr)^{-\mu}f\biggl(\frac{y-i}{-iy+1}\biggr), \\
\cC_\mu^{-1}&\colon\cH_\mu(\bH)\longrightarrow\cH_\mu(\bD), & (\cC_\mu^{-1} f)(z)&:=(iz+1)^{-\mu}f\biggl(\frac{z+i}{iz+1}\biggr), 
\end{align*}
where we choose the branches as $\bigl(\frac{-iy+1}{2}\bigr)^{-\mu}\bigr|_{y=i}=(iz+1)^{-\mu}|_{z=0}=1$. 
The norm $\Vert f\Vert_{\mu,\bD}$ is explicitly computed by using the Taylor expansion of $f$ as 
\[ \biggl\Vert \sum_{m=0}^\infty a_mz^m\biggr\Vert_{\mu,\bD}^2=\sum_{m=0}^\infty \frac{m!}{(\mu)_m}|a_m|^2, \]
where $(\mu)_m:=\mu(\mu+1)(\mu+2)\cdots(\mu+m-1)$ for $m\in\BZ_{>0}$ and $(\mu)_0:=1$. Then this is positive definite for $\mu>0$, 
and we can redefine $\cH_\mu(\bD)$, $\cH_\mu(\bH)$ for $\mu>0$ by 
\begin{align*}
\cH_\mu(\bD)&:=\biggl\{ f(z)=\sum_{m=0}^\infty a_mz^m\in\cO(\bD)\biggm| \Vert f\Vert_{\mu,\bD}^2:=\sum_{m=0}^\infty \frac{m!}{(\mu)_m}|a_m|^2<\infty\biggr\}, \\
\cH_\mu(\bH)&:=\cC_\mu(\cH_\mu(\bD)). 
\end{align*}
Also, for $\mu>0$ let 
\[ \cH_\mu(\bD)_{\text{fin}}:=\BC[z]\subset\cH_\mu(\bD), \qquad \cH_\mu(\bH)_{\text{fin}}:=\cC_\mu(\cH_\mu(\bD)_{\text{fin}}), \]
where $\BC[z]$ denotes the space of holomorphic polynomials on $\BC$. 

Next we consider the representations on $\cH_\mu(\bD)$, $\cH_\mu(\bH)$. When $\mu\in\BZ_{>0}$, the Lie group $SU(1,1)$ acts on $\cH_\mu(\bD)$ and $SL(2,\BR)$ acts on $\cH_\mu(\bH)$ unitarily by 
\[ \tau_\mu\biggl(\begin{bmatrix}a&b\\c&d\end{bmatrix}^{-1}\biggr)f(z):=(cz+d)^{-\mu}f\biggl(\frac{az+b}{cz+d}\biggr). \]
We note that this is not well-defined if $\mu\notin\BZ_{>0}$, since $(cz+d)^{-\mu}$ is not single-valued. 
However, we can make this well-defined by considering the universal covering groups of $SU(1,1)$ and $SL(2,\BR)$ for general $\mu>0$. 
The representations $(\tau_\mu,\cH_\mu(\bD))$, $(\tau_\mu,\cH_\mu(\bH))$ for $\mu>1$ are called \textit{holomorphic discrete series representations}, 
since they correspond to discrete spectra in $L^2(SU(1,1))$ and $L^2(SL(2,\BR))$. 
By differentiating $\tau_\mu$, we get the Lie algebra representation of $\mathfrak{su}(1,1)$ on $\cH_\mu(\bD)_{\text{fin}}$ and that of $\mathfrak{sl}(2,\BR)$ on $\cH_\mu(\bH)_{\text{fin}}$ as 
\[ \tau_\mu(X)f(z):=\frac{d}{dt}\tau_\mu(e^{tX})f(z)\biggr|_{t=0} \]
(we use the same symbol $\tau_\mu$), and by extending this complex-linearly, we get the representations of $\mathfrak{sl}(2,\BC)$ on $\cH_\mu(\bD)_{\text{fin}}$, $\cH_\mu(\bH)_{\text{fin}}$. 
The resulting representation is given by 
\begin{gather*}
\tau_\mu(H)f(z)=-2z\frac{df}{dz}(z)-\mu f(z), \qquad \tau_\mu(E)f(z)=-\frac{df}{dz}(z), \\ \tau_\mu(F)f(z)=z^2\frac{df}{dz}(z)+\mu zf(z). 
\end{gather*}
This is well-defined for all $\mu>0$. 
$\cH_\mu(\bD)_{\text{fin}}$ is $\mathfrak{su}(1,1)$-infinitesimally unitary and $\cH_\mu(\bH)_{\text{fin}}$ is $\mathfrak{sl}(2,\BR)$-infinitesimally unitary. 
When we distinguish $\bD$ and $\bH$, we write $\tau_\mu^\bD$, $\tau_\mu^\bH$ instead of $\tau_\mu$. These are related as 
\begin{align*}
\cC_\mu\circ\tau_\mu^\bD(g)&=\tau_\mu^\bH(\Int(C)g)\circ\cC_\mu && (g\in SU(1,1)), \\ 
\cC_\mu\circ\tau_\mu^\bD(X)&=\tau_\mu^\bH(\Ad(C)X)\circ\cC_\mu && (X\in\mathfrak{sl}(2,\BC)). 
\end{align*}
$\cH_\mu(\bD)$ has an orthogonal basis $\{z^m\}_{m=0}^\infty$, and $\cH_\mu(\bH)$ has an orthogonal basis 
$\bigl\{\cC_\mu(z^m)=\bigl(\frac{-iy+1}{2}\bigr)^{-\mu}\bigl(\frac{y-i}{-iy+1}\bigr)^m\bigr\}_{m=0}^\infty$. These satisfy 
\begin{subequations}\label{formula_basis}
\begin{equation}
\setlength{\belowdisplayskip}{0pt}
\langle z^m,z^n\rangle_{\mu,\bD}=\langle \cC_\mu(z^m),\cC_\mu(z^n)\rangle_{\mu,\bH}=\frac{m!}{(\mu)_m}\delta_{mn}, 
\end{equation}
\begin{align}
\setlength{\abovedisplayskip}{0pt}
\tau_\mu^\bD(H)z^m&=-(2m+\mu)z^m, & \tau_\mu^\bH({}^c\hspace{-1pt}H)\cC_\mu(z^m)&=-(2m+\mu)\cC_\mu(z^m), \\
\tau_\mu^\bD(E)z^m&=-mz^{m-1}, & \tau_\mu^\bH({}^c\hspace{-1pt}E)\cC_\mu(z^m)&=-m\cC_\mu(z^{m-1}), \\
\tau_\mu^\bD(F)z^m&=(m+\mu)z^{m+1}, & \tau_\mu^\bH({}^c\hspace{-1pt}F)\cC_\mu(z^m)&=(m+\mu)\cC_\mu(z^{m+1}). 
\end{align}
\end{subequations}
The representations $(\tau_\mu^\bD,\cH_\mu(\bD)_{\text{fin}})$, $(\tau_\mu^\bH,\cH_\mu(\bH)_{\text{fin}})$ for $\mu>0$ are called \textit{unitary highest weight representations}. 

Second, for $\mu>0$ we consider the Hilbert space and its dense subspace 
\begin{align*}
L^2_\mu(\BR^+)&:=\biggl\{ \phi\colon\BR^+\longrightarrow\BC \biggm| \Vert \phi\Vert_{\mu,\BR^+}^2:=\frac{2^\mu}{\Gamma(\mu)}\int_0^\infty |\phi(t)|^2t^{\mu-1}\,dt<\infty\biggr\}, \\
L^2_\mu(\BR^+)_{\text{fin}}&:=\BC[t]e^{-t}\subset L^2_\mu(\BR^+). 
\end{align*}
Let $\cL_\mu$ be the \textit{weighted Laplace transform} given by 
\begin{gather*}
\cL_\mu\colon L^2_\mu(\BR^+)\longrightarrow \cH_\mu(\bH), \qquad (\cL_\mu\phi)(y):=\frac{2^\mu}{\Gamma(\mu)}\int_0^\infty \phi(t)e^{iyt}t^{\mu-1}\,dt. 
\end{gather*}
For $m\in\BZ_{\ge 0}$ let 
\[ l_m^{(\mu)}(t):=i^m\sum_{j=0}^m\frac{(-m)_j}{(\mu)_jj!}(2t)^je^{-t}=\frac{i^m}{(\mu)_m}t^{-\mu+1}e^t\frac{d^m}{dt^m}t^{\mu+m-1}e^{-2t}\in L^2_\mu(\BR^+)_{\text{fin}}. \]
This is also written as $l_m^{(\mu)}(t)=i^m\frac{m!}{(\mu)_m}L_m^{(\mu-1)}(2t)e^{-t}$ by using the generalized Laguerre polynomial $L_m^{(\alpha)}(t)$. 
Then we can directly show that $\{l_m^{(\mu)}(t)\}_{m=0}^\infty$ forms an orthogonal basis of $L^2_\mu(\BR^+)$ satisfying 
\begin{gather*}
\langle l_m^{(\mu)}(t),l_n^{(\mu)}(t)\rangle_{\mu,\BR^+}=\frac{m!}{(\mu)_m}\delta_{mn}, \\
(\cL_\mu l_m^{(\mu)})(y)=i^m\sum_{j=0}^m\frac{(-m)_j}{j!}\biggl(\frac{1-iy}{2}\biggr)^{-\mu-j}
=i^m\biggl(\frac{1-iy}{2}\biggr)^{-\mu}\biggl(1-\frac{2}{1-iy}\biggr)^m=\cC_\mu(z^m)(y). 
\end{gather*}
Especially $\cL_\mu$ is a unitary isomorphism. Next we consider the following representation of $\mathfrak{sl}(2,\BC)$ on $L^2_\mu(\BR^+)_{\text{fin}}$. 
\begin{gather*}
\tau_\mu^{\BR^+}(H)\phi(t):=2t\frac{d\phi}{dt}(t)+\mu\phi(t), \qquad \tau_\mu^{\BR^+}(E)\phi(t):=-it\phi(t), \\ 
\tau_\mu^{\BR^+}(F):=-i\biggl(t\frac{d^2\phi}{dt^2}(t)+\mu\frac{d\phi}{dt}(t)\biggr). 
\end{gather*}
Then this representation is $\mathfrak{sl}(2,\BR)$-infinitesimally unitary, and $\cL_\mu$ intertwines the $\mathfrak{sl}(2,\BC)$-action, that is, we have 
\[ \cL_\mu\circ\tau_\mu^{\BR^+}(X)=\tau_\mu^{\bH}(X)\circ\cL_\mu \qquad (X\in\mathfrak{sl}(2,\BC)). \]
Especially we have 
\[ \tau_\mu^{\BR^+}({}^c\hspace{-1pt}H)l_m^{(\mu)}\!=-(2m+\mu)l_m^{(\mu)}, \quad \tau_\mu^{\BR^+}({}^c\hspace{-1pt}E)l_m^{(\mu)}\!=-ml_{m-1}^{(\mu)}, \quad
\tau_\mu^{\BR^+}({}^c\hspace{-1pt}F)l_m^{(\mu)}\!=(m+\mu)l_{m+1}^{(\mu)}. \]
Since $\tau_\mu^{\BR^+}({}^c\hspace{-1pt}H)=-i\tau_\mu^{\BR^+}(E-F)=-H_{\text{HO}}^{(\mu)}$ in (\ref{formula_HO_nu}), this solves the eigenvalue problem of $H_{\text{HO}}^{(\mu)}$ on $L^2_\mu(\BR^+)$. 
Now we have gotten three different pictures of the unitary highest weight representation of $\mathfrak{sl}(2,\BR)\simeq\mathfrak{su}(1,1)$, 
\[ \begin{matrix} L^2_\mu(\BR^+) & \overunderset{\cL_\mu}{\sim}{\longrightarrow} & \cH_\mu(\bH) & \overunderset{\cC_\mu}{\sim}{\longleftarrow} & \cH_\mu(\bD) \\
l_m^{(\mu)}(t) & \longmapsto & (\cL_\mu l_m^{(\mu)})(y)=(\cC_\mu z^m)(y) & \raisebox{5pt}{\rotatebox{180}{$\longmapsto$}} & z^m. \end{matrix} \]

\begin{remark}
By replacing $\BR^+$, $\bH$ and $\bD$ with a symmetric cone $\Omega\subset V$, a tube domain $\Omega+iV\subset V\otimes_\BR\BC$ and a bounded symmetric domain $D\subset V\otimes_\BR\BC$ 
associated to a Euclidean Jordan algebra $V$, these equivalences for $SL(2,\BR)\simeq SU(1,1)$ are generalized to those for Hermitian Lie groups of tube type, 
$Sp(n,\BR)$, $SU(n,n)$, $SO^*(4n)$, $SO_0(2,n)$ and $E_{7(-25)}$. See \cite{FK}. 
\end{remark}

Third, we consider a representation of $\mathfrak{sl}(2,\BC)$ on 
\[ L^2(\BR^n)\supset L^2(\BR^n)_{\text{fin}}:=\cP(\BR^n)e^{-r^2/2}, \]
given by 
\begin{gather*}
\tau^{\BR^n}(H)\psi(x)=\sum_{j=1}^n x_j\frac{\partial \psi}{\partial x_j}(x)+\frac{n}{2}\psi(x), \qquad \tau^{\BR^n}(E)\psi(x)=-\frac{i}{2}\sum_{j=1}^n x_j^2\psi(x), \\ 
\tau^{\BR^n}(F)\psi(x)=-\frac{i}{2}\sum_{j=1}^n\frac{\partial^2 \psi}{\partial x_j^2}(x). 
\end{gather*}
This representation is $\mathfrak{sl}(2,\BR)$-infinitesimally unitary. 
In terms of ${}^c\hspace{-1pt}H$, ${}^c\hspace{-1pt}E$, ${}^c\hspace{-1pt}F$, this representation is rewritten as 
\begin{align*}
\tau^{\BR^n}({}^c\hspace{-1pt}H)&=-i\tau^{\BR^n}(E-F)=-\frac{1}{2}\sum_{j=1}^n\biggl(x_j^2-\frac{\partial^2}{\partial x_j^2}\biggr)=-\sum_{j=1}^n\ra_j^\dagger\ra_j-\frac{n}{2}, \\
\tau^{\BR^n}({}^c\hspace{-1pt}E)&=\frac{1}{2}\tau^{\BR^n}(-iH\hspace{-1pt}+\hspace{-1pt}E\hspace{-1pt}+\hspace{-1pt}F)
=-\frac{i}{4}\sum_{j=1}^n\biggl(2x_j\frac{\partial}{\partial x_j}\hspace{-1pt}+\hspace{-1pt}1\hspace{-1pt}+\hspace{-1pt}x_j^2\hspace{-1pt}+\hspace{-1pt}\frac{\partial^2}{\partial x_j^2}\biggr)
=-\frac{i}{2}\sum_{j=1}^n (\ra_j)^2, \\
\tau^{\BR^n}({}^c\hspace{-1pt}F)&=\frac{1}{2}\tau^{\BR^n}(iH\hspace{-1pt}+\hspace{-1pt}E\hspace{-1pt}+\hspace{-1pt}F)
=-\frac{i}{4}\sum_{j=1}^n\biggl(-2x_j\frac{\partial}{\partial x_j}\hspace{-1pt}-\hspace{-1pt}1\hspace{-1pt}+\hspace{-1pt}x_j^2\hspace{-1pt}+\hspace{-1pt}\frac{\partial^2}{\partial x_j^2}\biggr)
=-\frac{i}{2}\sum_{j=1}^n (\ra_j^\dagger)^2. 
\end{align*}
According to the decomposition (\ref{formula_decomp}), for $\phi(t)\in L^2_{k+\frac{n}{2}}(\BR^+)_{\text{fin}}$, $h(x)\in \cH\cP_k(\BR^n)$, we have 
\[ \tau^{\BR^n}(X)\biggl(\phi\biggl(\frac{r^2(x)}{2}\biggr)h(x)\biggr)=(\tau_{k+\frac{n}{2}}^\mu(X)\phi)\biggl(\frac{r^2(x)}{2}\biggr)h(x) \qquad (X\in\mathfrak{sl}(2,\BC)). \]

\section{A generalization of NCHO}\label{section_NCHO}

In this section we consider a generalization of the non-commutative harmonic oscillator, with the Hamiltonian given by the self-adjoint operator on $L^2(\BR)\otimes\BC^2$, 
\begin{align*}
H_{\text{NCHO}}:\hspace{-3pt}&=\bA\biggl(-\frac{1}{2}\frac{d^2}{dx^2}+\frac{1}{2}x^2\biggr)+\bJ\biggl(x\frac{d}{dx}+\frac{1}{2}\biggr)+2\eta i\det(\bA\pm i\bJ)^{1/2}\bJ \\*
&=\bA\biggl(\ra^\dagger\ra+\frac{1}{2}\biggr)+\frac{1}{2}\bJ(\ra^2-(\ra^\dagger)^2)+2\eta i\det(\bA\pm i\bJ)^{1/2}\bJ, 
\end{align*}
where $\bA,\bJ\in M(2,\BR)$, $\bA={}^t\hspace{-1pt}\bA$, $\bJ=-{}^t\hspace{-1pt}\bJ$, and $\eta\in\BR$. 
That is, we want to find a function $\psi\in L^2(\BR^n)\otimes\BC^p$ satisfying 
\[ \biggl[-\bA_1\sum_{j=1}^n\frac{\partial^2}{\partial x_j^2}+i\bA_2\sum_{j=1}^n\biggl(x_j\frac{\partial}{\partial x_j}+\frac{1}{2}\biggr)+\bA_3\sum_{j=1}^n x_j^2\biggr]\psi=2\bC \psi, \]
where $\bA_k,\bC\in M(p,\BC)$, $\bA_k=\bA_k^\dagger$, $\bC=\bC^\dagger$. Here $\bC^\dagger={}^t\hspace{-1pt}\overline{\bC}$. 
Then by putting 
\begin{equation}\label{formula_AB}
\bA_1+\bA_3=:\bA, \qquad \frac{1}{2}(-i\bA_1+\bA_2+i\bA_3)=:\bB 
\end{equation}
so that $\bA,\bB\in M(p,\BC)$, $\bA=\bA^\dagger$, this is rewritten as 
\begin{equation}\label{formula_NCHO_a}
\biggl[\bA\sum_{j=1}^n\biggl(\ra_j^\dagger\ra_j+\frac{1}{2}\biggr)-i\bB\sum_{j=1}^n(\ra_j^\dagger)^2+i\bB^\dagger\sum_{j=1}^n(\ra_j)^2\biggr]\psi=2\bC \psi. 
\end{equation}
In terms of $\tau^{\BR^n}$, this becomes 
\[ \bigl[-\bA\tau^{\BR^n}({}^c\hspace{-1pt}H)+2\bB\tau^{\BR^n}({}^c\hspace{-1pt}F)-2\bB^\dagger\tau^{\BR^n}({}^c\hspace{-1pt}E)\bigr]\psi=2\bC \psi. \]
Then according to the decomposition (\ref{formula_decomp}), on $L^2_{k+\frac{n}{2}}(\BR^+)\otimes\cH\cP_k(\BR^n)\otimes\BC^p$, 
the above equation is reduced to an ordinary differential equation on $\cH_{k+\frac{n}{2}}(\bD)\otimes\BC^p$. 

\begin{theorem}
Let $h(x)\in\cH\cP_k(\BR^n)$, $\{u_m\}_{m=0}^\infty\subset\BC^p$, and let 
\[ \psi(x):=\sum_{m=0}^\infty u_m l_m^{(k+\frac{n}{2})}\biggl(\frac{r^2(x)}{2}\biggr)h(x). \]
Then $\psi(x)\in L^2(\BR^n)\otimes\BC^p$ and (\ref{formula_NCHO_a}) hold if and only if 
\begin{gather*}
f(z):=\sum_{m=0}^\infty u_m z^m\in \cH_{k+\frac{n}{2}}(\bD)\otimes\BC^p, \\
\bigl[-\bA\tau_{k+\frac{n}{2}}^{\bD}(H)+2\bB\tau_{k+\frac{n}{2}}^{\bD}(F)-2\bB^\dagger\tau_{k+\frac{n}{2}}^{\bD}(E)\bigr]f=2\bC f
\end{gather*}
hold. 
\end{theorem}

In the following, for general $\mu>0$ we consider 
\begin{equation}\label{formula_NCHO_D0}
\bigl[-\bA\tau_\mu^{\bD}(H)+2\bB\tau_\mu^{\bD}(F)-2\bB^\dagger\tau_\mu^{\bD}(E)\bigr]f=2\bC f 
\end{equation}
on $\cH_\mu(\bD)\otimes\BC^p$, that is, 
\[ \biggl[\bA\biggl(2z\frac{d}{dz}+\mu\biggr)+2\bB\biggl(z^2\frac{d}{dz}+\mu z\biggr)+2\bB^\dagger\frac{d}{dz}\biggr]f=2\bC f, \]
or equivalently, 
\begin{equation}\label{formula_NCHO_D}
(\bB z^2+\bA z+\bB^\dagger)\frac{df}{dz}=\Bigl(-\mu\bB z-\frac{\mu}{2}\bA+\bC\Bigr)f. 
\end{equation}

\begin{remark}\label{rem_confluence}
We consider the space of functions on $\sqrt{\mu}\bD:=\{z\in\BC\mid |z|<\sqrt{\mu}\}$, 
\[ \cH_\mu(\sqrt{\mu}\bD):=\biggl\{ f(w)\in\cO(\sqrt{\mu}\bD) \biggm| \frac{\mu-1}{\mu\pi}\int_{\sqrt{\mu}\bD}|f(w)|^2\biggl(1-\frac{|w|^2}{\mu}\biggr)^{\mu-2}\,dw<\infty\biggr\}, \]
and the map $\cH_\mu(\bD)\to\cH_\mu(\sqrt{\mu}\bD)$, $f(z)\mapsto f(w/\sqrt{\mu})$. 
Then the equation (\ref{formula_NCHO_D}) on $\cH_\mu(\bD)\otimes\BC^p$ is equivalent to the equation 
\[ \biggl(\frac{1}{\sqrt{\mu}}\bB w^2+\bA w+\sqrt{\mu}\bB^\dagger\biggr)\frac{df}{dw}=\Bigl(-\sqrt{\mu}\bB w-\frac{\mu}{2}\bA+\bC\Bigr)f \]
on $\cH_\mu(\sqrt{\mu}\bD)\otimes\BC^p$. Then by putting $\sqrt{\mu}\bB=:\widetilde{\bB}$, $\bC-\frac{\mu}{2}\bA=:\widetilde{\bC}$ and taking the limit $\mu\to\infty$, this ``converges'' to the equation 
\[ (\bA w+\widetilde{\bB}^\dagger)\frac{df}{dw}=(-\widetilde{\bB} w+\widetilde{\bC})f \]
on 
\[ \cF(\BC)\otimes\BC^p:=\biggl\{ f(w)\in\cO(\BC) \biggm| \frac{1}{\pi}\int_{\BC}|f(w)|^2 e^{-|w|^2}\,dw<\infty\biggr\}\otimes\BC^p, \]
and via the Bargmann transform, this is equivalent to the equation 
\[ (\bA \ra^\dagger\ra+\widetilde{\bB}^\dagger\ra+\widetilde{\bB} \ra^\dagger)\psi=\widetilde{\bC}\psi \]
on $L^2(\BR)\otimes\BC^p$. Especially, let $p=2$, let $\bsigma_1,\bsigma_2,\bsigma_3\in M(2,\BC)$ be the Pauli matrices, let $\bI\in M(2,\BC)$ be the identity, and let $\omega,\Delta,g,\varepsilon\in\BR$. 
If $\bA=\omega\bI$, $\widetilde{\bB}=g\bsigma_1$ and $\widetilde{\bC}=-\Delta\bsigma_3-\varepsilon\bsigma_1+\lambda\bI$, then this is the asymmetric quantum Rabi model \cite{B1, B2, KRW, XZBL}, 
\begin{equation}\label{formula_Rabi}
(\omega\bI \ra^\dagger\ra+g\bsigma_1(\ra+\ra^\dagger)+\Delta\bsigma_3+\varepsilon\bsigma_1)\psi=\lambda \psi. 
\end{equation}
That is, the non-commutative harmonic oscillator is regarded as a ``covering model'' of the quantum Rabi model, as is pointed out in \cite{RW, W}. 
Similarly, if $\bA=\omega\bI$, $\widetilde{\bB}=g(\bsigma_1-i\bsigma_2)/2=:g\bsigma^-$, $\widetilde{\bB}^\dagger=g(\bsigma_1+i\bsigma_2)/2=:g\bsigma^+$ and $\widetilde{\bC}=-\Delta\bsigma_3+\lambda\bI$, 
then this is the Jaynes--Cummings model \cite{JC}, 
\[ (\omega\bI \ra^\dagger\ra+g(\bsigma^+\ra+\bsigma^-\ra^\dagger)+\Delta\bsigma_3)\psi=\lambda \psi. \]
\end{remark}

We return to the analysis of (\ref{formula_NCHO_D}). The inverse of $\bB z^2+\bA z+\bB^\dagger$ is written as 
\[ (\bB z^2+\bA z+\bB^\dagger)^{-1}=\frac{\adj(\bB z^2+\bA z+\bB^\dagger)}{\det(\bB z^2+\bA z+\bB^\dagger)}, \]
where $\adj$ is the adjugate matrix. Its denominator and numerator have degrees at most $2p$ and $2p-2$ respectively. 
Now we assume that all poles of $(\bB z^2+\bA z+\bB^\dagger)^{-1}$ are of order 1, and let 
\begin{equation}\label{partfrac}
(\bB z^2+\bA z+\bB^\dagger)^{-1}=\sum_{j=1}^N \frac{1}{z-\alpha_j}\bP_{\alpha_j}+\sum_{k=0}^M \bQ_k z^k 
\end{equation}
be its partial fraction decomposition, where $\alpha_j\in\BC$ are distinct, and $\bP_{\alpha_j},\bQ_k\in M(p,\BC)$, $\bP_{\alpha_j}\ne 0$. Then the following holds. 

\begin{lemma}\label{lem_partfrac}
\begin{enumerate}
\item If $\alpha\ne 0$ is a pole of $(\bB z^2+\bA z+\bB^\dagger)^{-1}$, then $\overline{\alpha}^{-1}$ is also a pole of it. 
\item For a pole $\alpha\ne 0$, $\bP_{\overline{\alpha}^{-1}}=-\bP_\alpha^\dagger$, and $\bQ_k=\bzero$. 
\item If $\det(\bB)\ne 0$, then $\displaystyle \sum_{j=1}^N \bP_{\alpha_j}=\bzero$,\; $\displaystyle \sum_{j=1}^N \alpha_j \bP_{\alpha_j}\bB=\bI$. 
\item If $\det(\bB)=0$, then $\displaystyle \sum_{j=1}^N \bP_{\alpha_j}=\bP_0^\dagger$,\; $\displaystyle \sum_{j=1}^N \bP_{\alpha_j}\bB=\bzero$,\; 
$\displaystyle \sum_{j=1}^N \alpha_j \bP_{\alpha_j}\bB=\bI-\bP_0^\dagger\bA$. 
\item $\rank \bP_{\alpha_j}\le \dim\Ker(\alpha_j^2\bB+\alpha_j\bA+\bB^\dagger)$. 
\item $\bP_{\alpha_j}(2\alpha_j\bB+\bA)\bP_{\alpha_j}=\bP_{\alpha_j}$. 
\end{enumerate}
\end{lemma}
\begin{proof}
1. If $\alpha\ne 0$ is a pole of $(\bB z^2+\bA z+\bB^\dagger)^{-1}$, then $\alpha^2\bB+\alpha\bA+\bB^\dagger$ is non-invertible, and hence 
\[ \overline{\alpha}^{-2}\bB+\overline{\alpha}^{-1}\bA+\bB^\dagger=\overline{\alpha}^{-2}(\alpha^2\bB+\alpha\bA+\bB^\dagger)^\dagger \]
is also non-invertible. Therefore $\overline{\alpha}^{-1}$ is also a pole of $(\bB z^2+\bA z+\bB^\dagger)^{-1}$. 

2. By (\ref{partfrac}) we have 
\begin{align*}
&(\bB z^2+\bA z+\bB^\dagger)^{-1}=z^{-2}((\bB\overline{z}^{-2}+\bA\overline{z}^{-1}+\bB^\dagger)^{-1})^\dagger \\
&=z^{-2}\biggl(\sum_{j=1}^N \frac{1}{z^{-1}-\overline{\alpha_j}}\bP_{\alpha_j}^\dagger+\sum_{k=0}^M \bQ_k^\dagger z^{-k}\biggr) 
=\sum_{j=1}^N \biggl(\frac{1}{z}+\frac{-\overline{\alpha_j}}{\overline{\alpha_j}z-1}\biggr)\bP_{\alpha_j}^\dagger+\sum_{k=0}^M \bQ_k^\dagger z^{-k-2}. 
\end{align*}
Then by the simple-pole assumption, we have $\bQ_k=\bzero$, and comparing the residues at $\alpha_j=\overline{\alpha_{j'}}^{-1}$ with (\ref{partfrac}), we get 
\[ \bP_\alpha=-\bP_{\overline{\alpha}^{-1}}^\dagger, \qquad \sum_{j=1}^N\bP_{\alpha_j}=\begin{cases} \bzero & (0\text{ is not a pole}), \\ \bP_0^\dagger & (0\text{ is a pole}). \end{cases} \]

3, 4. By (\ref{partfrac}) we have 
\[ \sum_{j=1}^N \frac{1}{z(z-\alpha_j)}\bP_{\alpha_j}(\bB z^2+\bA z+\bB^\dagger)=\frac{1}{z}\bI, \]
and taking the limit $z\to\infty$, we get $\sum_{j=1}^N \bP_{\alpha_j}\bB=\bzero$. Similarly, we have 
\begin{align*}
\bI&=\sum_{j=1}^N \frac{1}{z-\alpha_j}\bP_{\alpha_j}(\bB z^2+\bA z+\bB^\dagger) \\
&=\sum_{j=1}^N\biggl(\biggl(1+\frac{\alpha_j}{z-\alpha_j}\biggr)\bP_{\alpha_j}(\bB z+\bA)+\frac{1}{z-\alpha_j}\bP_{\alpha_j}\bB^\dagger\biggr) \\
&=\sum_{j=1}^N\biggl(\bP_{\alpha_j}\bB z+\frac{\alpha_j z}{z-\alpha_j}\bP_{\alpha_j}\bB+\bP_{\alpha_j}\bA+\frac{1}{z-\alpha_j}\bP_{\alpha_j}(\alpha_j\bA+\bB^\dagger)\biggr). 
\end{align*}
Then since $\sum_{j=1}^N \bP_{\alpha_j}\bB=\bzero$, by taking the limit $z\to\infty$, we have 
\[ \sum_{j=1}^N\alpha_j\bP_{\alpha_j}\bB=\bI-\sum_{j=1}^N\bP_{\alpha_j}\bA=\begin{cases} \bI & (0\text{ is not a pole}), \\ \bI-\bP_0^\dagger\bA & (0\text{ is a pole}). \end{cases} \]

5. By comparing the residues at $z=\alpha_j$ of 
\[ \sum_{j=1}^N \frac{1}{z-\alpha_j}(\bB z^2+\bA z+\bB^\dagger)\bP_{\alpha_j}=\bI, \]
we get $(\alpha_j^2\bB+\alpha_j\bA+\bB^\dagger)\bP_{\alpha_j}=\bzero$, and the desired formula holds. 

6. By differentiating (\ref{partfrac}) with respect to $z$, we get 
\begin{align*}
-\sum_{j=1}^N \frac{1}{(z-\alpha_j)^2}\bP_{\alpha_j}&=-(\bB z^2+\bA z+\bB^\dagger)^{-1}(2\bB z+\bA)(\bB z^2+\bA z+\bB^\dagger)^{-1} \\
&=-\sum_{j=1}^N\sum_{k=1}^N \frac{1}{(z-\alpha_j)(z-\alpha_k)}\bP_{\alpha_j}(2\bB z+\bA)\bP_{\alpha_k}. 
\end{align*}
Then multiplying $(z-\alpha_j)^2$ and substituting $z=\alpha_j$, we get the desired formula. 
\end{proof}

From this lemma we easily get the following. 
\begin{theorem}\label{thm_Fuchs}
Suppose that all poles of $(\bB z^2+\bA z+\bB^\dagger)^{-1}$ are of order 1, and we consider the partial fraction decomposition (\ref{partfrac}). 
\begin{enumerate}
\item The equation (\ref{formula_NCHO_D}) is equal to the Fuchsian equation 
\begin{equation}\label{formula_Fuchs}
\frac{df}{dz}=\sum_{j=1}^N\frac{1}{z-\alpha_j}\bP_{\alpha_j}\Bigl(-\mu\Bigl(\alpha_j\bB+\frac{1}{2}\bA\Bigr)+\bC\Bigr)f. 
\end{equation}
\item If (\ref{formula_NCHO_D}) is singular at $z=\alpha$ $(\ne 0,\infty)$, then it is also singular at $z=\overline{\alpha}^{-1}$. 
\item We have 
$\displaystyle -\sum_{j=1}^N \bP_{\alpha_j}\Bigl(-\mu\Bigl(\alpha_j\bB+\frac{1}{2}\bA\Bigr)+\bC\Bigr)
=\begin{cases} \mu\bI & (\det(\bB)\ne 0), \\ \mu\bI-\bP_0^\dagger\bigl(\frac{\mu}{2}\bA+\bC\bigr) & (\det(\bB)=0). \end{cases}$
\item We have $\rank\bigl(\bP_{\alpha_j}\bigl(-\mu\bigl(\alpha_j\bB+\frac{1}{2}\bA\bigr)+\bC\bigr)\bigr)\le \dim\Ker(\alpha_j^2\bB+\alpha_j\bA+\bB^\dagger)$. 
That is, $\bP_{\alpha_j}\bigl(-\mu\bigl(\alpha_j\bB+\frac{1}{2}\bA\bigr)+\bC\bigr)$ has an eigenvalue 0 
with multiplicity at least $p-\dim\Ker(\alpha_j^2\bB+\alpha_j\bA+\bB^\dagger)$. 
\item For $u\in \operatorname{Im}\bP_{\alpha_j}\subset\BC^p$, we have 
$\bP_{\alpha_j}\bigl(-\mu\bigl(\alpha_j\bB+\frac{1}{2}\bA\bigr)+\bC\bigr)u=\bP_{\alpha_j}\bC u-\frac{\mu}{2}u$. 
That is, non-zero eigenvalues of $\bP_{\alpha_j}\bigl(-\mu\bigl(\alpha_j\bB+\frac{1}{2}\bA\bigr)+\bC\bigr)$ are equal to $-\frac{\mu}{2}$ plus those of $\bP_{\alpha_j}\bC$. 
\end{enumerate}
\end{theorem}

Suppose that all singularities $\alpha_j$ are not in $S^1:=\{z\in\BC\mid |z|=1\}$. Then if a solution $f$ of (\ref{formula_NCHO_D}) is holomorphic at 
all the singularities $\alpha_1,\ldots,\alpha_{N'}$ inside the unit disk $\bD$, then $f$ automatically has the radius of convergence strictly larger than 1 at the origin, 
and hence $f$ belongs to $\cH_\mu(\bD)\otimes\BC^p$. That is, we do not need to estimate the norm of $f$. Especially, if $\bA,\bB$ (or $\bA_1,\bA_2,\bA_3$ in (\ref{formula_AB})) satisfy the condition 
\begin{multline}\label{cond_pos_def}
\bB z+\bA+\bB^\dagger\overline{z}=2(\bA_1\cos^2\theta+\bA_2\cos\theta\sin\theta+\bA_3\sin^2\theta) \\ \text{ is positive definite for all }z=ie^{-2i\theta}\in S^1, 
\end{multline}
then a solution $f$ of (\ref{formula_NCHO_D}) in $\cO(\bD)\otimes\BC^p$ automatically belongs $\cH_\mu(\bD)\otimes\BC^p$. 

\begin{remark}
If the left hand side of (\ref{formula_NCHO_D0}) is bounded below on $\cH_\mu(\bD)\otimes\BC^p$, then $\bB z+\bA+\bB^\dagger\overline{z}$ is positive semidefinite for all $z\in S^1$. Indeed, 
\begin{align*}
&\biggl\langle\bigl[-\bA\tau_\mu^{\bD}(H)+2\bB\tau_\mu^{\bD}(F)-2\bB^\dagger\tau_\mu^{\bD}(E)\bigr]\sum_{m=0}^\infty u_mz^m, \sum_{m=0}^\infty v_mz^m\biggr\rangle_{\mu,\bD} \\
&=\sum_{m=0}^\infty \frac{(m+1)!}{(\mu)_m}\biggl(\frac{m+1}{\mu+m}v_{m+1}^\dagger \bA u_{m+1}+2v_{m+1}^\dagger \bB u_m+2v_m^\dagger \bB^\dagger u_{m+1}+\frac{\mu+m}{m+1}v_m^\dagger \bA u_m\biggr)
\end{align*}
holds for all $\{u_m\},\{v_m\}\subset\BC^p$ by (\ref{formula_basis}), and especially if the left hand side of (\ref{formula_NCHO_D0}) has the lower bound $c\in\BR$, 
then for $u\in\BC^p$, $w\in \bD$, we have 
\begin{align*}
&0\le\bigl\langle\bigl[-\bA\tau_\mu^{\bD}(H)+2\bB\tau_\mu^{\bD}(F)-2\bB^\dagger\tau_\mu^{\bD}(E)-c\bI\bigr](1-z\overline{w})^{-\mu}u, (1-z\overline{w})^{-\mu}u\bigr\rangle_{\mu,\bD} \\
&=\biggl\langle\bigl[-\bA\tau_\mu^{\bD}(H)+2\bB\tau_\mu^{\bD}(F)-2\bB^\dagger\tau_\mu^{\bD}(E)-c\bI\bigr]\sum_{m=0}^\infty \frac{(\mu)_m}{m!}\overline{w}^mz^mu, 
\sum_{m=0}^\infty \frac{(\mu)_m}{m!}\overline{w}^mz^mu\biggr\rangle_{\mu,\bD} \\
&=\sum_{m=0}^\infty\frac{(\mu)_{m+1}}{m!}|w|^{2m}u^\dagger(\bA |w|^2+2\bB w+2\bB^\dagger\overline{w}+\bA)u-c\sum_{m=0}^\infty\frac{(\mu)_m}{m!}|w|^{2m}u^\dagger u \\
&=\mu(1-|w|^2)^{-\mu-1}u^\dagger(\bA |w|^2+2\bB w+2\bB^\dagger\overline{w}+\bA)u-c(1-|w|^2)^{-\mu}u^\dagger u. 
\end{align*}
Then by taking the limit $|w|\nearrow 1$, we get $u^\dagger(\bB w+\bA+\bB^\dagger\overline{w})u\ge 0$ for all $w\in S^1$. 
Also, by putting $w=0$ we get $\bA\ge\frac{c}{\mu}\bI$. 
\end{remark}

Next we consider the action of the Lie group $SU(1,1)$. In the following we assume $\mu\in\BZ_{>0}$, but we may consider general $\mu>0$ by taking the universal covering group of $SU(1,1)$. 
We take an arbitrary $g\in SU(1,1)$. Then the equation (\ref{formula_NCHO_D0}) is transformed as 
\begin{align*}
&\tau_\mu^\bD(g)\bigl[-\bA\tau_\mu^{\bD}(H)+2\bB\tau_\mu^{\bD}(F)-2\bB^\dagger\tau_\mu^{\bD}(E)\bigr]\tau_\mu^\bD(g)^{-1}f \\
&=\bigl[-\bA\tau_\mu^{\bD}(\Ad(g)H)+2\bB\tau_\mu^{\bD}(\Ad(g)F)-2\bB^\dagger\tau_\mu^{\bD}(\Ad(g)E)\bigr]f=2\bC f 
\end{align*}
on $\cH_\mu(\bD)\otimes\BC^p$, where $\Ad(g)X:=gXg^{-1}$. Let ${}^g\!\bA,{}^g\bB\in M(p,\BC)$ be the matrices satisfying 
\begin{align*}
&-\bA\otimes \Ad(g)H+2\bB\otimes \Ad(g)F-2\bB^\dagger\otimes \Ad(g)E \\
&=-{}^g\!\bA\otimes H+2{}^g\bB\otimes F-2{}^g\bB^\dagger\otimes E
\in \mathfrak{sl}(2,\BC)\otimes M(p,\BC), 
\end{align*}
namely, for $g=\bigl[\begin{smallmatrix}a&b\\ \overline{b}&\overline{a}\end{smallmatrix}\bigr]\in SU(1,1)$, let 
\[ ({}^g\!\bA,{}^g\bB):=\bigl((|a|^2+|b|^2)\bA-2(\overline{a}b\bB+a\overline{b}\bB^\dagger), -\overline{ab}\bA+\overline{a}^2\bB+\overline{b}^2\bB^\dagger\bigr). \]
Then this equation is equivalent to 
\[ ({}^g\bB z^2+{}^g\!\bA z+{}^g\bB^\dagger)\frac{df}{dz}=\Bigl(-\mu{}^g\bB z-\frac{\mu}{2}{}^g\!\bA+\bC\Bigr)f. \]
On the other hand, the equation (\ref{formula_Fuchs}) is transformed as 
\[ \tau_\mu^\bD(g)\frac{d}{dz}\tau_\mu^\bD(g)^{-1}f
=\tau_\mu^\bD(g)\sum_{j=1}^N\frac{1}{z-\alpha_j}\bP_{\alpha_j}\Bigl(-\mu\Bigl(\alpha_j\bB+\frac{1}{2}\bA\Bigr)+\bC\Bigr)\tau_\mu^\bD(g)^{-1}f, \]
and by direct computation this is equivalent to 
\begin{align*}
\frac{df}{dz}&=\biggl(\sum_{\substack{1\le j\le N \\ g.\alpha_j\ne\infty}}\frac{1}{z-g.\alpha_j}\bP_{\alpha_j}\Bigl(-\mu\Bigl(\alpha_j\bB+\frac{1}{2}\bA\Bigr)+\bC\Bigr) \\
&\hspace{42.5pt} -\frac{1}{z-g.\infty}\biggl(\sum_{1\le j\le N}\bP_{\alpha_j}\Bigl(-\mu\Bigl(\alpha_j\bB+\frac{1}{2}\bA\Bigr)+\bC\Bigr)+\mu\bI\biggr)\biggr)f \\
&=\biggl(\sum_{\substack{1\le j\le N \\ g.\alpha_j\ne\infty}}\frac{1}{z-g.\alpha_j}
\bP_{\alpha_j}\Bigl(-\mu\Bigl(\alpha_j\bB+\frac{1}{2}\bA\Bigr)+\bC\Bigr)-\frac{1}{z-g.\infty}\bP_0^\dagger\Bigl(\frac{\mu}{2}\bA+\bC\Bigr)\biggr)f, 
\end{align*}
where we have used Theorem \ref{thm_Fuchs}.3, and we set $\bP_0:=\bzero$ if $\det(\bB)\ne 0$. 
Here, for $g=\bigl[\begin{smallmatrix}a&b\\ \overline{b}&\overline{a}\end{smallmatrix}\bigr]\in SU(1,1)$ we write 
\[ g.\alpha:=\begin{cases} \displaystyle \frac{a\alpha+b}{\overline{b}\alpha+\overline{a}} & \bigl(\alpha\ne -\frac{\overline{a}}{\overline{b}}\bigr), \\ 
\infty & \bigl(\alpha= -\frac{\overline{a}}{\overline{b}}\bigr), \end{cases}
\qquad g.\infty:=\frac{a}{\overline{b}}. \]
Then these two equations must coincide. That is, we have 
\begin{align*}
\frac{df}{dz}&=({}^g\bB z^2+{}^g\!\bA z+{}^g\bB^\dagger)^{-1}\Bigl(-\mu{}^g\bB z-\frac{\mu}{2}{}^g\!\bA+\bC\Bigr)f \\
&=\biggl(\sum_{\substack{1\le j\le N \\ g.\alpha_j\ne\infty}}\frac{1}{z-g.\alpha_j}
\bP_{\alpha_j}\Bigl(-\mu\Bigl(\alpha_j\bB+\frac{1}{2}\bA\Bigr)+\bC\Bigr)-\frac{1}{z-g.\infty}\bP_0^\dagger\Bigl(\frac{\mu}{2}\bA+\bC\Bigr)\biggr)f. 
\end{align*}
Also, the positivity condition (\ref{cond_pos_def}) is stable under $(\bA,\bB)\mapsto ({}^g\!\bA,{}^g\bB)$, 
since for $z\in S^1$, $g=\bigl[\begin{smallmatrix}a&b\\\overline{b}&\overline{a}\end{smallmatrix}\bigr]\in SU(1,1)$, we have 
\[ {}^g\bB z+{}^g\!\bA+{}^g\bB^\dagger\overline{z}=\lvert-\overline{b}z+a\rvert^2\bigl(\bB(g^{-1}.z)+\bA+\bB^\dagger\overline{(g^{-1}.z)}\bigr). \]

In the following we consider the case $p=2$. We assume $(\bA,\bB)$ satisfies the positivity condition (\ref{cond_pos_def}), 
and that $\det(\bB z^2+\bA z+\bB^\dagger)$ has at least 3 distinct roots. 
If it has 4 distinct roots, then by (\ref{cond_pos_def}) the roots are not on $S^1$, and by Lemma \ref{lem_partfrac}.1, 
they are of the form $\{\beta,\gamma,\overline{\beta}^{-1},\overline{\gamma}^{-1}\}$ with $\beta,\gamma\in\bD$. 
Then by considering $g\in SU(1,1)$ satisfying $g.\beta=0$, $g.\overline{\beta}^{-1}=\infty$, $g.\gamma=:\alpha\in\bD$ and replacing $(\bA,\bB)$ with $({}^g\!\bA,{}^g\bB)$, 
without loss of generality we may assume $\det(\bB)=0$ and $\det(\bB z^2+\bA z+\bB^\dagger)$ has the roots $\{0,\alpha,\overline{\alpha}^{-1}\}$ with $\alpha\in\bD$. 
Next, since $\bA$ is positive definite by (\ref{cond_pos_def}), by replacing $\bB$ and $\bC$ 
with $\bA^{-\frac{1}{2}}\bB\bA^{-\frac{1}{2}}$ and $\bA^{-\frac{1}{2}}\bC\bA^{-\frac{1}{2}}$, 
we may assume $\bA=\bI$. In addition, by taking a gauge transform by a unitary matrix, we may assume 
\begin{equation}\label{std_form}
\bA=\begin{bmatrix} 1&0\\0&1 \end{bmatrix}, \quad \bB=\begin{bmatrix} b_1&b_2\\0&0 \end{bmatrix}, \quad \bC=\begin{bmatrix} c_1&c_2\\\overline{c_2}&c_3 \end{bmatrix} \quad
\biggl(\begin{matrix} c_1,c_3\in\BR,\; b_1,b_2,c_2\in\BC, \\ b_1\ne 0,\; 2|b_1|+|b_2|^2<1 \end{matrix}\biggr). 
\end{equation}
Then (\ref{formula_NCHO_D}) becomes 
\begin{equation}\label{formula_NCHO_rank2}
\begin{bmatrix} b_1z^2+z+\overline{b_1} & b_2z^2 \\ \overline{b_2} & z \end{bmatrix}\frac{df}{dz}
=\begin{bmatrix} -\mu b_1z-\frac{\mu}{2}+c_1 & -\mu b_2z+c_2 \\ \overline{c_2} & -\frac{\mu}{2}+c_3 \end{bmatrix}. 
\end{equation}
Let 
\begin{align*}
&\begin{bmatrix} b_1z^2+z+\overline{b_1} & b_2z^2 \\ \overline{b_2} & z \end{bmatrix}^{-1} 
\begin{bmatrix} -\mu b_1z-\frac{\mu}{2}+c_1 & -\mu b_2z+c_2 \\ \overline{c_2} & -\frac{\mu}{2}+c_3 \end{bmatrix} \\
&=\frac{1}{z(b_1z^2+(1-|b_2|^2)z+\overline{b_1})} \begin{bmatrix} z & -b_2z^2 \\ -\overline{b_2} & b_1z^2+z+\overline{b_1} \end{bmatrix} 
\begin{bmatrix} -\mu b_1z-\frac{\mu}{2}+c_1 & -\mu b_2z+c_2 \\ \overline{c_2} & -\frac{\mu}{2}+c_3 \end{bmatrix} \\
&=:\frac{1}{z}\begin{bmatrix} 0 & 0 \\ X_0 & Y_0 \end{bmatrix}+\frac{1}{(z-\alpha)(z-\overline{\alpha}^{-1})}\begin{bmatrix} V_1z+V_2 & W_1z+W_2 \\ X_1z+X_2 & Y_1z+Y_2 \end{bmatrix}
\\&=\frac{1}{z}\begin{bmatrix} 0 & 0 \\ X_0 & Y_0 \end{bmatrix}+\frac{1}{z-\alpha_+}\begin{bmatrix} V_+ & W_+ \\ X_+ & Y_+ \end{bmatrix}
+\frac{1}{z-\alpha_-}\begin{bmatrix} V_- & W_- \\ X_- & Y_- \end{bmatrix}, 
\end{align*}
where we write $\alpha=:\alpha_+$, $\overline{\alpha}^{-1}=:\alpha_-$, so that 
\begin{gather*}
\alpha_{\pm}=\frac{-1+|b_2|^2\pm\sqrt{(1-|b_2|^2)^2-4|b_1|^2}}{2b_1}, \quad 
\begin{bmatrix} 0 & 0 \\ X_0 & Y_0 \end{bmatrix}=\frac{1}{\overline{b_1}}\begin{bmatrix} 0 & 0 \\ -\overline{b_2} & \overline{b_1} \end{bmatrix}
\begin{bmatrix} c_1-\frac{\mu}{2} & c_2 \\ \overline{c_2} & c_3-\frac{\mu}{2} \end{bmatrix}, \\
\begin{bmatrix} V_\pm & W_\pm \\ X_\pm & Y_\pm \end{bmatrix}=\frac{\pm 1}{b_1(\alpha_+-\alpha_-)}\begin{bmatrix} 1 & -b_2\alpha_\pm \\ -\overline{b_2}\alpha_\pm^{-1} & |b_2|^2 \end{bmatrix}
\begin{bmatrix} -\mu b_1\alpha_\pm-\frac{\mu}{2}+c_1 & -\mu b_2\alpha_\pm+c_2 \\ \overline{c_2} & -\frac{\mu}{2}+c_3 \end{bmatrix}. 
\end{gather*}
We note that $V_\pm Y_\pm-W_\pm X_\pm=0$ holds by Theorem \ref{thm_Fuchs}.4. 
Then (\ref{formula_NCHO_rank2}) is rewritten as 
\begin{align*}
\frac{df_1}{dz}&=\frac{V_1z+V_2}{(z-\alpha)(z-\overline{\alpha}^{-1})}f_1+\frac{W_1z+W_2}{(z-\alpha)(z-\overline{\alpha}^{-1})}f_2, \\
\frac{df_2}{dz}&=\biggl(\frac{X_0}{z}+\frac{X_1z+X_2}{(z-\alpha)(z-\overline{\alpha}^{-1})}\biggr)f_1+\biggl(\frac{Y_0}{z}+\frac{Y_1z+Y_2}{(z-\alpha)(z-\overline{\alpha}^{-1})}\biggr)f_2. 
\end{align*}
By transforming this into a single higher-order differential equation, we get 
\begin{align*}
\biggl(\frac{d}{dz}-\frac{Y_0}{z}-\frac{Y_1z+Y_2}{(z-\alpha)(z-\overline{\alpha}^{-1})}\biggr)\frac{(z-\alpha)(z-\overline{\alpha}^{-1})}{W_1z+W_2}
\biggl(\frac{d}{dz}-\frac{V_1z+V_2}{(z-\alpha)(z-\overline{\alpha}^{-1})}\biggr)f_1& \\
=\biggl(\frac{X_0}{z}+\frac{X_1z+X_2}{(z-\alpha)(z-\overline{\alpha}^{-1})}\biggr)f_1&, 
\end{align*}
that is, 
\begin{align*}
&\biggl[ \frac{d^2}{dz^2}-\biggl(\frac{Y_0}{z}+\frac{V_++Y_+-1}{z-\alpha}+\frac{V_-+Y_--1}{z-\overline{\alpha}^{-1}}+\frac{W_1}{W_1z+W_2}\biggr)\frac{d}{dz} \\
&{}+\hspace{-1pt}\frac{V_1(Y_0\hspace{-1pt}+\hspace{-1pt}Y_1)\hspace{-1pt}-\hspace{-1pt}W_1(X_0\hspace{-1pt}+\hspace{-1pt}X_1)}{(z-\alpha)(z-\overline{\alpha}^{-1})} 
\hspace{-1pt}+\hspace{-1pt}\frac{V_2Y_0-W_2X_0}{z(z\hspace{-1pt}-\hspace{-1pt}\alpha)(z\hspace{-1pt}-\hspace{-1pt}\overline{\alpha}^{-1})}
\hspace{-1pt}+\hspace{-1pt}\frac{V_2W_1-V_1W_2}{(z\hspace{-1pt}-\hspace{-1pt}\alpha)(z\hspace{-1pt}-\hspace{-1pt}\overline{\alpha}^{-1})(W_1z\hspace{-1pt}+\hspace{-1pt}W_2)} \biggr]f_1=0.
\end{align*}
Therefore we get the following. 

\begin{proposition}
The equation (\ref{formula_NCHO_rank2}) is equivalent to the single differential equation 
\begin{align*}
&\biggl[ \frac{d^2}{dz^2}+\biggl(\frac{-\kappa_0+\frac{\mu}{2}}{z}+\frac{1-\kappa_1+\frac{\mu}{2}}{z-\alpha}+\frac{1+\overline{\kappa_1}+\frac{\mu}{2}}{z-\overline{\alpha}^{-1}}
-\frac{1}{z-\epsilon}\biggr)\frac{d}{dz} \\*
&{}+\frac{\mu\bigl(-\overline{\kappa_0}+\frac{\mu}{2}\bigr)}{(z-\alpha)(z-\overline{\alpha}^{-1})} +\frac{q_1}{z(z-\alpha)(z-\overline{\alpha}^{-1})}
+\frac{q_2}{(z-\alpha)(z-\overline{\alpha}^{-1})(z-\epsilon)} \biggr]f_1=0, 
\end{align*}
where 
\begin{gather*}
\alpha=\frac{-1+|b_2|^2+\sqrt{(1-|b_2|^2)^2-4|b_1|^2}}{2b_1}, \qquad 
\kappa_0=\frac{\overline{b_1}c_3-\overline{b_2}c_2}{\overline{b_1}}=\frac{\tr(\adj(\bB^\dagger)\bC)}{\tr(\adj(\bA)\bB^\dagger)}, \\
\kappa_1=\frac{1}{b_1(\alpha-\overline{\alpha}^{-1})}\tr\biggl(\begin{bmatrix} 1 & -b_2\alpha \\ -\overline{b_2}\alpha^{-1} & |b_2|^2 \end{bmatrix}
\begin{bmatrix} c_1 & c_2 \\ \overline{c_2} & c_3 \end{bmatrix}\biggr)=\frac{\tr(\adj(\bB\alpha+\bA+\bB^\dagger\alpha^{-1})\bC)}{(\alpha-\overline{\alpha}^{-1})\tr(\adj(\bA)\bB)}, \\
\epsilon=\frac{c_2}{b_2\bigl(c_3+\frac{\mu}{2}\bigr)}, \qquad q_1=\frac{\bigl(c_1-\frac{\mu}{2}\bigr)\bigl(c_3-\frac{\mu}{2}\bigr)-|c_2|^2}{b_1}
=\frac{\det\bigl(\bC-\frac{\mu}{2}\bA\bigr)}{\tr(\adj(\bA)\bB)}, \\
q_2=\frac{\mu\bigl(b_2\bigl(c_1-\frac{\mu}{2}\bigr)-b_1c_2\bigr)+b_2\bigl(\bigl(c_1-\frac{\mu}{2}\bigr)\bigl(c_3-\frac{\mu}{2}\bigr)-|c_2|^2\bigr)}{b_1b_2\bigl(c_3+\frac{\mu}{2}\bigr)}. 
\end{gather*}
\end{proposition}

This equation has original regular singularities at $\{0,\alpha,\overline{\alpha}^{-1},\infty\}$ and an additional apparent singularity at $z=\epsilon$. 
The characteristic exponents at each singularity are given as 
\[ \begin{Bmatrix} z=0 & z=\alpha & z=\overline{\alpha}^{-1} & z=\epsilon & z=\infty \\ 0 & 0 & 0 & 0 & \mu \\ 
1+\kappa_0-\frac{\mu}{2} & \kappa_1-\frac{\mu}{2} & -\overline{\kappa_1}-\frac{\mu}{2} & 2 & -\overline{\kappa_0}+\frac{\mu}{2} \end{Bmatrix}. \]
Such an equation appears in the context of isomonodromic deformations and Painlev\'e equations (see e.g. \cite{Ok}). 
If the space of solutions for the equation (\ref{formula_NCHO_rank2}) in $\cH_\mu(\bD)\otimes\BC^2$ is 2-dimensional, 
then all local solutions around the singularities $0,\alpha\in\bD$ must be holomorphic, and this holds only if the corresponding characteristic exponents are positive integers. That is, 

\begin{corollary}
If the space of solutions for the equation (\ref{formula_NCHO_rank2}) in $\cH_\mu(\bD)\otimes\BC^2$ is 2-dimensional, then $1+\kappa_0-\frac{\mu}{2},\kappa_1-\frac{\mu}{2}\in\BZ_{>0}$. 
\end{corollary}
We note that this condition is not sufficient since there may exist solutions with logarithmic terms. 

As a special case, if $\bB=\bB^\dagger$ holds, or equivalently if $\bA_1=\bA_3$ holds in (\ref{formula_AB}), 
then by a suitable $g\in SO(1,1):=SU(1,1)\cap \allowbreak M(2,\BR)$ and by a suitable gauge transform, 
we can change $(\bA,\bB,\bC)$ to a form in (\ref{std_form}) with $b_1\in\BR$, $b_2=0$. Then by transforming (\ref{formula_NCHO_rank2}) into a single higher-order differential equation, 
we get a Heun equation \cite{R}, 
\begin{equation}\label{HeunDE}
\biggl[ \frac{d^2}{dz^2}+\biggl(\frac{-\kappa_0+\frac{\mu}{2}}{z}+\frac{1-\kappa_1+\frac{\mu}{2}}{z-\alpha}+\frac{1+\kappa_1+\frac{\mu}{2}}{z-\alpha^{-1}}\biggr)\frac{d}{dz}
+\frac{\mu(1-\kappa_0+\frac{\mu}{2})z+q_1}{z(z-\alpha)(z-\alpha^{-1})} \biggr]f_1=0, 
\end{equation}
with regular singularities at $\{0,\alpha,\alpha^{-1},\infty\}$ and with no additional apparent singularity. 
The characteristic exponents at each singularity are given as 
\[ \begin{Bmatrix} z=0 & z=\alpha & z=\alpha^{-1} & z=\infty \\ 0 & 0 & 0 & \mu \\
1+\kappa_0-\frac{\mu}{2} & \kappa_1-\frac{\mu}{2} & -\kappa_1-\frac{\mu}{2} & 1-\kappa_0+\frac{\mu}{2} \end{Bmatrix}. \]

\begin{example}
The original $\eta$-shifted non-commutative harmonic oscillator $H_{\mathrm{NCHO}}\psi\allowbreak=\lambda\psi$ treated in \cite{RW}, in the standard form, corresponds to (\ref{formula_NCHO_a}) with $n=1$, $p=2$ and 
\begin{gather*}
\bA=\begin{bmatrix} \beta & 0 \\ 0 & \gamma \end{bmatrix}, \qquad \bB=\frac{1}{2}\begin{bmatrix} 0 & i \\ -i & 0 \end{bmatrix}, \qquad
2\bC=\begin{bmatrix} \lambda & 2\eta i\sqrt{\beta\gamma-1} \\ -2\eta i\sqrt{\beta\gamma-1} & \lambda \end{bmatrix}. 
\end{gather*}
The positivity condition (\ref{cond_pos_def}) is equivalent to $\beta,\gamma>0$, $\beta\gamma>1$, which makes the spectra of $H_{\mathrm{NCHO}}$ discrete and bounded below. 
Its restriction to $L^2_{\mathrm{even}}(\BR)\otimes\BC^2$ and $L^2_{\mathrm{odd}}(\BR)\otimes\BC^2$ are equivalent to the differential equation (\ref{formula_NCHO_D}) 
with $\mu=\frac{1}{2}$ and $\mu=\frac{3}{2}$ respectively. Let $l:=\Bigl[\begin{smallmatrix} \sqrt{\beta i}^{-1}&0\\0&\sqrt{-\gamma i}^{-1}\end{smallmatrix}\Bigr]\in GL(2,\BC)$, and we put 
\begin{gather*}
\bA':=l\bA l^\dagger=\begin{bmatrix}1&0\\0&1\end{bmatrix}, \qquad \bB':=l\bB l^\dagger=\frac{1}{2\sqrt{\beta\gamma}}\begin{bmatrix}0&1\\1&0\end{bmatrix}=\begin{bmatrix}0&g\\g&0\end{bmatrix}, \\
\bC':=l\bC l^\dagger=\frac{\lambda}{2}\begin{bmatrix}\beta^{-1}&0\\0&\gamma^{-1}\end{bmatrix}+\eta\frac{\sqrt{\beta\gamma-1}}{\sqrt{\beta\gamma}}\begin{bmatrix}0&1\\1&0\end{bmatrix}
=\begin{bmatrix}\lambda'-\Delta&-\varepsilon\\-\varepsilon&\lambda'+\Delta\end{bmatrix},
\end{gather*}
with $g,\lambda',\Delta,\varepsilon\in\BR$, $|g|<\frac{1}{2}$, as an analogue of the Rabi model (\ref{formula_Rabi}) (with $\omega=1$). 
Let $\tanh(\theta):=\sqrt{\beta\gamma}^{-1}=2g$ so that (\ref{formula_NCHO_D}) has 5 regular singularities at $\pm\tanh(\theta/2)^{\pm 1}$ and $\infty$, 
and let $h_+:=\Bigl[\begin{smallmatrix} \cosh(\theta/2) & \sinh(\theta/2) \\ \sinh(\theta/2) & \cosh(\theta/2) \end{smallmatrix}\Bigr]$, 
$h_-:=\Bigl[\begin{smallmatrix} i\cosh(\theta/2) & -i\sinh(\theta/2) \\ i\sinh(\theta/2) & -i\cosh(\theta/2) \end{smallmatrix}\Bigr]\in SU(1,1)$. 
Then by replacing $\bA,\bB,\bC$ with 
\begin{align*}
{}^{h_{\pm}}\!\bA'&=(\cosh\theta)\bA'\mp 2(\sinh\theta)\bB'
=\frac{1}{\sqrt{1-4g^2}}\begin{bmatrix} 1 & \mp 4g^2 \\ \mp 4g^2 & 1 \end{bmatrix}, \\
{}^{h_{\pm}}\bB'&=-\frac{1}{2}(\sinh\theta)\bA'\pm(\cosh\theta)\bB'
=\frac{g}{\sqrt{1-4g^2}}\begin{bmatrix} -1 & \pm 1 \\ \pm 1 & -1 \end{bmatrix}, 
\end{align*}
and $\bC'$, we get equivalent differential equations. Transforming this into a single differential equation, 
we get the Heun differential equation (\ref{HeunDE}) with 
\begin{gather*}
\alpha=2g=\sqrt{\beta\gamma}^{-1}, \quad \kappa_0=\kappa_\pm, \quad \kappa_1=\kappa_\mp, \quad q_1=q_\pm, \qquad \text{where} \\ 
\kappa_\pm=\frac{\lambda'\mp\varepsilon}{\sqrt{1-4g^2}}=\frac{\lambda}{4}\frac{\beta+\gamma}{\sqrt{\beta\gamma(\beta\gamma-1)}}\pm\eta, \\ 
\begin{split}
q_\pm&=-\frac{1}{2g}\biggl(\lambda^{\prime 2}-\frac{\lambda'\mu}{\sqrt{1-4g^2}}+\frac{\mu^2}{4}(1+4g^2)\pm\frac{4g^2\mu\varepsilon}{\sqrt{1-4g^2}}-\varepsilon^2-\Delta^2\biggr) \\
&=-\frac{1}{\sqrt{\beta\gamma}}\biggl(\frac{\lambda^2}{4}-\frac{\lambda\mu\sqrt{\beta\gamma}(\beta+\gamma)}{4\sqrt{\beta\gamma-1}}+\frac{\mu^2}{4}(\beta\gamma+1)\mp\eta\mu-\eta^2(\beta\gamma-1)\biggr). 
\end{split}
\end{gather*}
For details on analysis of this equation, see \cite{RW}. 

Next we consider a confluence process as in Remark \ref{rem_confluence}. Let $\tilde{g}:=\sqrt{\mu}g$, $\tilde{\lambda}:=\lambda'-\frac{\mu}{2}$, $z=w/\sqrt{\mu}$. 
Then 
by taking the limit $\mu\to\infty$, we get 
\begin{gather*}
\biggl[ \frac{d^2}{dw^2}+\biggl(\frac{-\tilde{\kappa}_\pm}{w}+\frac{1-\tilde{\kappa}_\mp}{w-2\tilde{g}}-2\tilde{g}\biggr)\frac{d}{dw}
-\frac{2\tilde{g}(1-\tilde{\kappa}_\pm)w-\tilde{q}_\pm}{w(w-2\tilde{g})} \biggr]f_1=0, \quad \text{where} \\
\tilde{\kappa}_\pm:=\tilde{\lambda}+\tilde{g}^2\mp\varepsilon, \qquad 
\tilde{q}_\pm:=(\tilde{\lambda}+\tilde{g}^2)(\tilde{\lambda}-3\tilde{g}^2)\pm 4\tilde{g}^2\varepsilon-\varepsilon^2-\Delta^2, 
\end{gather*}
since $\kappa_\pm=\tilde{\kappa}_\pm+\frac{\mu}{2}+O(\mu^{-1})$, $q_\pm=-\frac{\sqrt{\mu}}{2\tilde{g}}(\tilde{q}_\pm+O(\mu^{-1}))$ hold. 
This gives the confluent Heun picture of the asymmetric quantum Rabi model (\ref{formula_Rabi}). For details see \cite{KRW, RW}. 
\end{example}

\end{document}